\documentclass[runningheads]{llncs}

\usepackage{makeidx}  
\usepackage{graphicx}
\usepackage{color}
\usepackage{lineno,hyperref}

\usepackage{amsfonts}
\usepackage{amsmath}
\usepackage{amsthm} 
\usepackage{enumerate}
\usepackage{makecell}
\usepackage{appendix}
\newtheorem{rrule}{Reduction-Rule} 
\newtheorem{xrule}{Rule} 

\begin{document}
\title{An Improved Kernel and Parameterized Algorithm for Almost Induced Matching}
\titlerunning{Almost Induced Matching}
%
\author{
    Yuxi Liu 
    \and
    Mingyu Xiao\orcidID{0000-0002-1012-2373}
}
\authorrunning{Y. Liu and M. Xiao}
%
\institute{University of Electronic Science and Technology of China, Chengdu, China
\email{202211081321@std.uestc.edu.cn, myxiao@uestc.edu.cn}}
\maketitle              

\begin{abstract}
    An induced subgraph is called an induced matching if each vertex is a degree-1 vertex in the subgraph.
    The \textsc{Almost Induced Matching} problem asks whether we can delete at most $k$ vertices from the input graph such that the remaining graph is an induced matching.
    This paper studies parameterized algorithms for this problem by taking the size $k$ of the deletion set as the parameter.
    First, we prove a $6k$-vertex kernel for this problem, improving the previous result of $7k$. Second, we give an $O^*(1.6765^k)$-time and polynomial-space algorithm, improving the previous running-time bound of $O^*(1.7485^k)$.

\end{abstract}
\section{Introduction}

An \textit{induced matching} is an induced regular graph of degree 1.
The problem of finding an induced matching of maximum size is known as \textsc{Maximum Induced Matching} (MIM),
which is crucial in algorithmic graph theory.
In many graph classes such as trees \cite{golumbic2000new}, chordal graphs \cite{cameron1989induced}, circular-arc graphs \cite{golumbic1993irredundancy}, and interval graphs \cite{golumbic2000new},
a maximum induced matching can be found in polynomial time.
However, \textsc{Maximum Induced Matching} is NP-hard in planar 3-regular graphs or planar bipartite graphs with degree-2 vertices in one part and degree-3 vertices in the other part \cite{duckworth2005approximability,ko2003bipartite,stockmeyer1982np}.
Kobler and Rotics \cite{kobler2003finding} proved the NP-hardness of this problem in Hamiltonian graphs, claw-free graphs, chair-free graphs, line graphs, and regular graphs.
The applications of induced matchings are diverse and include secure communication channels, VLSI design, and network flow problems, as demonstrated by Golumbic and Lewenstein \cite{golumbic2000new}.

In terms of exact algorithms, Gupta, Raman, and Saurabh \cite{gupta2012maximum} demonstrated that \textsc{Maximum Induced Matching} can be solved in $O^*(1.6957^n)$ time.
This result was later improved to $O^*(1.3752^n)$ by Xiao and Tan \cite{xiao2017exact}.
For subcubic graphs (i.e., such graphs with maximum degree $3$), Hoi, Sabili and Stephan \cite{DBLP:journals/corr/abs-2201-03220} showed that \textsc{Maximum Induced Matching} can be solved in $O^*(1.2630^n)$ time and polynomial space.
In terms of parameterized algorithms, where the parameter is the solution size $k'$, \textsc{Maximum Induced Matching} is W[1]-hard in general graphs \cite{moser2009parameterizedregular}, and is not expected to have a polynomial kernel.
However, Moser and Sidkar \cite{moser2009parameterizedinducedmatching} showed that the problem becomes fixed-parameter tractable (FPT) when the graph is a planar graph, by providing a linear-size problem kernel.
The kernel size was improved to $40k'$ by Kanj et al. \cite{kanj2011induced}.
In this paper, we study parameterized algorithms for \textsc{Maximum Induced Matching} with the parameter being the number $k$ of vertices not in the induced matching.
The problem is formally defined as follows:

\noindent\rule{\linewidth}{0.2mm}
\textsc{Almost Induced Matching}\\
\textbf{Instance:} A graph $G=(V,E)$ and an integer $k$.\\
\textbf{Question:}  Is there a vertex subset $S\subseteq V$ of size at most $k$ whose deletion makes the graph an induced matching?\\
\rule{\linewidth}{0.2mm}

\textsc{Almost Induced Matching} becomes FPT when we take the size $k$ of the deletion set as the parameter.
Xiao and Kou \cite{xiao2016almost,xiao2020parameterized} showed that \textsc{Almost Induced Matching} can be solved in $O^*(1.7485^k)$. The same running time bound was also achieved in \cite{kumar2020deletion} recently.
In terms of kernelization, Moser and Thilikos \cite{moser2009parameterizedregular} provided a kernel of $O(k^3)$ vertices. Then, Mathieson and Szeider \cite{mathieson2012editing} improved the result to $O(k^2)$.
Last, Xiao and Kou \cite{xiao2016almost,xiao2020parameterized} obtained the first linear-vertex kernel of $7k$ vertices for this problem.
In this paper, we first improve the kernel size to $6k$ vertices, and then
give an $O^*(1.6765^k)$-time and polynomial-space algorithm for \textsc{Almost Induced Matching}.

For kernelization, the main technique in this paper is a variant of the crown decomposition.
We find a maximal 3-path packing $\mathcal{P}$ in the graph,
partition the vertex set into two parts $V(\mathcal{P})$ and $Q=V\setminus V(\mathcal{P})$,
and reduce the number of the size-1 connected components of $G[Q]$ to at most $k$.
The size-2 components in $G[Q]$ can be reduced to at most $k$ by using the new ``AIM crown decomposition" technique.
Note that each connected component of $G[Q]$ has a size of at most $2$.
The size of $\mathcal {P }$ is at most $k$ since at least one vertex must be deleted from each $3$-path.
In the worst case scenario, when there are $k$ 3-paths in $\mathcal{P}$, $k$ size-2 components, and $k$ size-1 components in $G[Q]$,
the graph has at most $3k + 2k + k = 6k$ vertices.

Our parameterized algorithm is a branch-and-search algorithm.
We first handle vertices with degrees of 1 and 2, followed by those with degree at least 5.
Next, we only need to deal with degree-3 and degree-4 vertices.
The different part from previous algorithms is as follows. We use refined rules to deal with degree-2, degree-3, and degree-4 vertices by carefully checking the local structures. Therefore, we can avoid previous bottlenecks.


\section{Preliminaries}
In this paper, we only consider simple and undirected graphs.
Let $G=(V, E)$ be a graph with $n=|V|$ vertices and $m=|E|$ edges.
A singleton $\{v\}$ may be denoted as $v$.
We use $V(G')$ and $E(G')$ to denote the vertex set and edge set of a graph $G'$, respectively.
A vertex $v$ is called a \emph{neighbor} of a vertex $u$ if there is an edge $\{u,v\} \in E$.
Let $N(v)$ denote the set of neighbors of $v$. For a vertex subset $X$, let $N(X)=\cup_{v\in X}N(v)\setminus X$ and $N[X]=N(X)\cup X$.
We use $d(v)=|N(v)|$ to denote the degree of a vertex $v$ in $G$.
A vertex of degree $d$ is called a degree-$d$ vertex.
For a vertex subset $X\subseteq V$, the subgraph induced by $X$ is denoted by $G[X]$, and $G[V\setminus X]$ is also written as $G\setminus X$ or $G-X$.
A vertex in a vertex subset $X$ is called an \emph{$X$-vertex}.
Two vertex-disjoint subgraphs $X_1$ and $X_2$ are \emph{adjacent} if there is an edge $\{u,v\} \in E$ with $u\in X_1$ and $v\in X_2$.
A graph is called an \emph{induced matching} if the size of each connected component in it is two.
A vertex subset $S$ is called an \emph{AIM-deletion set} of $G$ if $G\setminus S$ is an induced matching.

A \emph{$3$-path} $P_3 = \{u_1,u_2,u_3\}$ is a path with two edges $\{u_1,u_2\}$ and $\{u_2,u_3\}$.
Two $3$-paths $L_1$ and $L_2$ are \emph{vertex-disjoint} if $V(L_{1}) \cap V(L_{2}) = \emptyset$.
A set of 3-paths $\mathcal{P} = \{L_1,L_2,...,L_t\}$ is called a \emph{$P_3$-packing} if any two 3-paths in it are vertex-disjoint.
A $P_3$-packing is \emph{maximal} if there is no $P_3$-packing $\cal P'$ such that $|\mathcal{P}|<|\mathcal{P'}|$ and $\mathcal{P} \subset \mathcal{P'}$.
We also use $V(\mathcal{P})$ to denote the set of vertices in 3-paths in a $P_3$-packing $\mathcal{P}$.

We will use a variant of the classic \emph{VC crown decomposition}.
Now, we give the definition of VC crown decomposition \cite{DBLP:conf/alenex/Abu-KhzamCFLSS04,chor2004linear} for the ease of reference.


\begin{definition}\label{VC-decomposition}
    A \textit{VC crown decomposition} of a graph $G = (V, E)$ is a partition $(C, H, R)$ of the vertex set $V$ satisfying the following properties.
    \begin{enumerate}
        \item There is no edge between $C$ and $R$.
        \item $C$ is an independent set.
        \item There is an injective mapping (matching) $M: H\rightarrow C$ such that $\{x,M(x)\}\in E$ holds for all $x\in H$.
    \end{enumerate}
\end{definition}


\begin{lemma}[\cite{chor2004linear}]\label{VC-lemma}
    If graph $G = (V, E)$ has an independent set $I \subseteq V$ with $|I| > |N(I)|$, then a VC crown decomposition $(C, H, R)$ with $\emptyset \neq  C \subseteq I$ and $H \subseteq N(I)$ can be found in linear time.
\end{lemma}

\section{Kernelization}

In this section, we show that \textsc{Almost Induced Matching} allows a kernel of $6k$ vertices.
The main idea of our algorithm is as follows.
The first step of this algorithm is to find a maximal $P_3$-packing $\mathcal {P}$ in $G=(V,E)$ by using a greedy method.
For a \textbf{yes}-instance, we will have $|\mathcal{P}| \leq k$.
We partition the vertex set $V$ into two parts  $P=V (\mathcal{P})$ and $Q = V\setminus P$. Each connected component in $G[Q]$ is of size at most 2 by the maximality of $\mathcal{P}$.
We will bound the number of components in $G[Q]$ to bound the size of $Q$.


Let $Q_0$ denote the set of degree-0 vertices in $G[Q]$, and $Q_1$ denote the set of degree-1 vertices in $G[Q]$.
Use \emph{$Q_1$-edge} to denote the components with size 2 in $G[Q]$.
For each $L_i \in \mathcal{P}$, let $Q(L_i)$ denote the set of $Q$-vertices in the components of $G[Q]$ adjacent to $L_i$.
Let $V_i$ denote $Q(L_i)\cup V(L_i)$.
It should be noted that a vertex in $Q(L_i)$ might not be adjacent to any vertex in $L_i$.

A 3-path $L_i\in \mathcal{P}$ is \emph{good} if at most one vertex in $L_i$ is adjacent to $Q$-vertices.
A 3-path $L_i\in \mathcal{P}$ is \emph{bad} if at least two vertices in $L_i$ are adjacent to $Q$-vertices.
A maximal $P_3$-packing is \emph{proper} if it holds that $|V_i| \leq 6$ for any bad 3-path $L_i$ in it.

In our kernelization algorithm, we will first find an arbitrary maximal $P_3$-packing $\mathcal{P}$. After this, we will use two rules in \cite{xiao2020parameterized}
to update $\mathcal{P}$.

\begin{xrule} \label{rule_rule1}
If there is a 3-path $L_i \in \mathcal{P}$ such that  $G[V_i]$ contains at least two vertex-disjoint 3-paths, then replace $L_i$ by these 3-paths in $\mathcal{P}$ to increase the size of $\mathcal{P}$ by at least one.
\end{xrule}

\begin{lemma}[\cite{xiao2020parameterized}]\label{size7}
    Assume that Rule~\ref{rule_rule1} can not be applied on the current instance. For any $L_i \in \mathcal{P}$ with $|V_i|\geq 7$,
    there is a 3-path $L'_i$ in $G[V_i]$ such that $L'_i$ is a good 3-path after replacing $L_i$ with $L'_i$ in $\mathcal{P}$.
    Furthermore, the 3-path $L'_i$ can be found in constant time.
\end{lemma}

We call the 3-path $L'_i$ in Lemma~\ref{size7} a \emph{quasi-good 3-path}.
\begin{xrule}\label{rule_rule2}
    For any 3-path $L_i \in \mathcal{P}$ with $|V_i|\geq 7$, if it is not good, then replace $L_i$ with a quasi-good 3-path $L'_i$ in $\mathcal{P}$.
\end{xrule}



\begin{lemma}[\cite{xiao2020parameterized}] \label{lem_proper0}
For any initially maximal $P_3$-packing $\mathcal{P}$ in $G$, we can apply Rules~\ref{rule_rule1} and~\ref{rule_rule2} in $O(n^2)$ time to change $\mathcal{P}$ to a proper $P_3$-packing.
\end{lemma}


This lemma implies that the number of $Q$-vertices in the components of $G[Q]$ adjacent to bad 3-paths in $\mathcal {P}$ is small.
So we only need to bound the number of $Q_0$-vertices and $Q_1$-edges adjacent to good 3-paths in $\mathcal {P}$.
For any proper $P_3$-packing $\mathcal {P}$ in $G$, we can bound the number of $Q_0$-vertices adjacent to good 3-paths by the following Lemma \ref*{prop2} from \cite{xiao2020parameterized}.

\begin{lemma}[\cite{xiao2020parameterized}]\label{prop2}
    Let $\mathcal {P }$ be a proper $P_3$-packing in graph $G $, where the number of bad 3-paths is $x $.
    If $(G, k)$ is a \textbf{yes}-instance, then the number of $Q_0 $-vertices only adjacent to good 3-paths (not adjacent to bad 3-paths) in $\mathcal {P} $ is at most $k-x$.
\end{lemma}

The main contribution in this paper is to bound the number of $Q_1$-edges adjacent to good 3-paths.
We need to use the following technique called \emph{AIM crown decomposition}.

\begin{definition}[AIM crown decomposition]\label{def-decomposition}
    An \textit{AIM crown decomposition} of a graph $G=(V,E)$ is a decomposition $(C, H, R) $ of the vertex set $V$ such that
    \begin{enumerate}
        \item there is no edge between $C $ and $R$;
        \item the induced subgraph $G[C] $ is an induced matching;
        \item there is an injective mapping (matching from vertices to edges) $M: H \rightarrow E(C)$ such that for all $v\in H$, there exists $u\in M(v)$ such that $(u,v)\in E$.

    \end{enumerate}
\end{definition}

Fig. \ref*{Fig:1} illustrates an AIM crown decomposition.
We have the following Lemma \ref*{AIM-crown-lemma} for AIM crown decomposition, which allows us to find parts of the solution based on an AIM crown decomposition.

\begin{figure}[!t]
    \centering
    \includegraphics[scale=0.20]{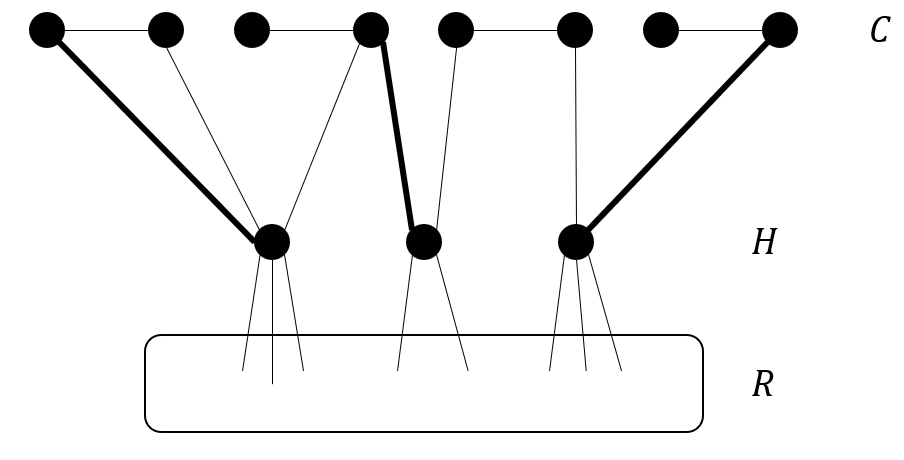}
    \caption{An AIM crown decomposition $(C, H, R)$, where the three bold edges form an injective matching from $H$ to $C$.}
    \label{Fig:1}
\end{figure}

\begin{lemma}\label{AIM-crown-lemma}
    Let $(C, H, R)$ be an AIM crown decomposition of a graph $G$.
    There is a minimum AIM-deletion set $S$ such that $H\subseteq S$.
\end{lemma}

\begin{proof}
    Let $S^*$ be an arbitrary minimum AIM-deletion set of $G$.
    Let $G'=G\setminus S^* $. Note that $G'$ is an induced matching.
    We partition $H$ into two sets $H_1$ and $H_2$, where $H_1 = H \cap S^*$ and $H_2 = H \setminus H_1$.
    So all vertices in $H_2$ are in the induced matching $G'$.
    We further partition $H_2$ into two sets $H_C $ and $H_R$, where each vertex in $H_C$ is adjacent to a vertex in $C \cup H$ in the induced matching $G'$ and each vertex in $H_R$ is adjacent to a vertex in $R $ in the induced matching $G'$.
    Let $N_{G\setminus S^*}(H_R)$ be the set of vertices adjacent to vertices in $H_R$ in the induced matching $G'$.


    We partition $C$ into two sets $C_1  $ and $C_2$, where two vertices of each edge in $C_1$ are both in $G'$ and $C_2 = C \setminus C_1 $.
    Let $T = H_2\cup C_2\cup N_{G\setminus S^*}(H_R)$.
    See Fig. \ref*{Fig:2} for an illustration.

    \begin{figure}[t]
        \centering
        \includegraphics[scale=0.23]{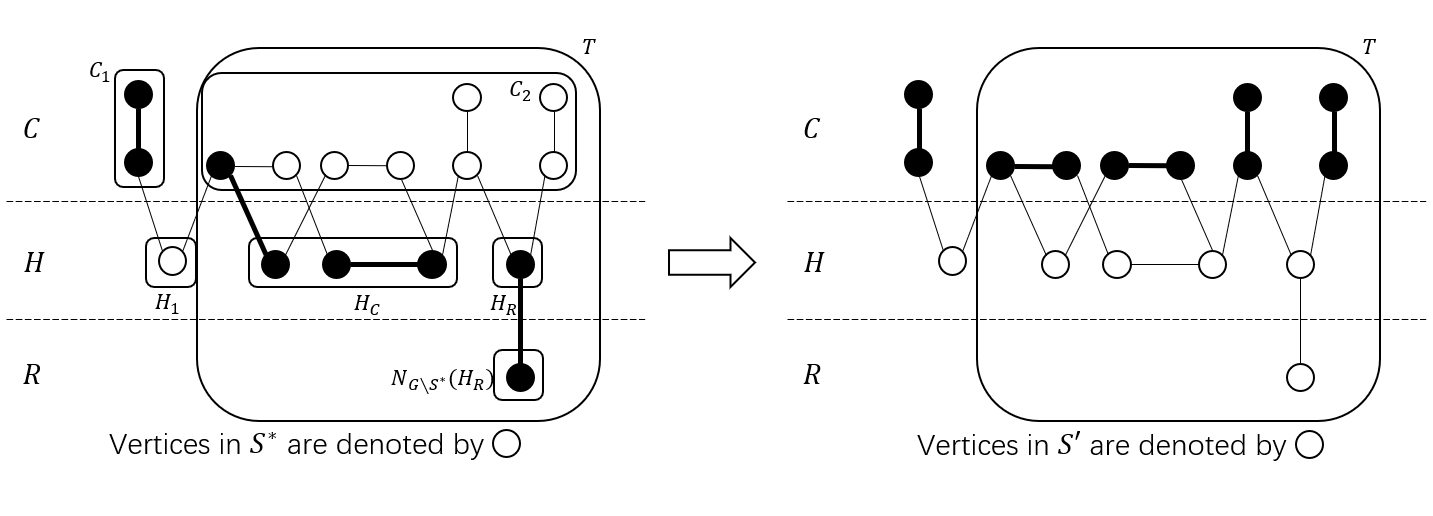}
        \caption{Sets $S^*$ and $S'$ in the proof of Lemma \ref*{AIM-crown-lemma}.}
        \label{Fig:2}
    \end{figure}

    Let $S'=(S^*\setminus C)\cup H_2 \cup N_{G\setminus S^*}(H_R)$.
    Firstly, we show that $S'$ is still an AIM-deletion set of $G$.
    See Fig. \ref*{Fig:2} for an illustration of $S^*$ and $S'$.
    It is easy to see that $H\subseteq S'$.

    Let $S_1=S^* \setminus (H\cup C)$.
    We can see that $S_1$ is an AIM-deletion set of the induced subgraph $G\setminus (H\cup C \cup N_{G\setminus S^*}(H_R))$ by definition of $S^* $ and $N_{G\setminus S^*}(H_R)$.
    So $S_1\cup H\cup N_{G\setminus S^*}(H_R)$ is an AIM-deletion of origin graph $G $.
    We can see that $S_1\cup H\cup N_{G\setminus S^*}(H_R) = S'$.
    Then, we show that $|S'|\leq |S^*|$ to finish our proof.

    By definition of $S'$, we have that
    \begin{equation}
        |S'|=|S^*|-|S^*\cap C|+|H_2|+|N_{G\setminus S^*}(H_R)|. %
    \end{equation}

    Note that, $H_2\cap S^*=\emptyset,N_{G\setminus S^*}(H_R) \cap S^* =\emptyset, C_1 \cap S^* = \emptyset$ because of their definitions. We have that
    \begin{equation}
        |S^*\cap T|=|S^*\cap C_2|=|S^*\cap C|. %
    \end{equation}

    Recall the definition of $T = H_2\cup C_2\cup N_{G\setminus S^*}(H_R)$.
    Since $H$ is an $C$-$R$ cutset, vertices in $C_2 $ can not be adjacent to $V\setminus T $ in $G'$, and $H_2, N_{G\setminus S^*}(H_R)$ can not be adjacent to $V\setminus T $ in $G' $ by their definitions.
    We can claim that there are at most $|H_2|$ induced matchings in $G'[T\setminus S^*]$.
    We get that
    \begin{equation}
        |S^*\cap T|\ge|T|-2|H_2|. %
    \end{equation}

    Note that all edges in $C$ adjacent to $H_2$ must be in $C_2$, so we have $\forall v\in H_2, M(v)\in C_2$, so $|C_2|\ge|H_2|,|V(C_2)|\ge 2|H_2|$.
    By the definition of $T $, we have that
    \begin{equation}
        |H_2|+|N_{G\setminus S^*}(H_R)|=|T|-|V(C_2)|\leq |T|-2|H_2|. %
    \end{equation}

    Equality (2) with inequalities (3) and (4) together imply that
    \begin{equation}
        |H_2|+|N_{G\setminus S^*}(H_R)|\leq |S^*\cap C|. %
    \end{equation}

    Equality (1) with inequality (5) together imply that
    \[
        |S'|=|S^*|-|S^*\cap C|+|H_2|+|N_{G\setminus S^*}(H_R)|\leq|S^*|-|S^*\cap C|+|S^*\cap C|=|S^*|. %
    \]

\end{proof}

Once given an AIM crown decomposition $(C, H, R)$ of the graph, we can reduce the instance by including $H$ to the deletion set and removing $H \cup C$ from the graph by Lemma \ref*{AIM-crown-lemma}.
Next, we show that when the number of $Q_1$-edges is large, we can always find an AIM crown decomposition in polynomial time.

\begin{lemma}\label{AIM-crown-lemma-2}
    Let $G = (V, E)$ be a graph with each connected component containing more than two vertices, and $A$ and $B\subseteq V$ be two disjoint vertex sets such that\\
    (i) no vertex in $A$ is adjacent to a vertex in $V\setminus (A \cup B)$;\\
    (ii) the induced subgraph $G[A]$ is an induced matching.\\
    If $|A| > 2|B|$, then the graph allows an AIM crown decomposition $(C, H, R)$ with $\emptyset \neq C \subseteq A$ and $H \subseteq B$, and the AIM crown decomposition $(C, H, R)$ can be found in polynomial time.
\end{lemma}

\begin{proof}
   Firstly, we construct an auxiliary bipartite graph $G' = (V', E')$ with $V' = (A', B')$ as follows:
   Each vertex in $A'$ corresponds to a component (an edge) in the induced matching $G[A]$, and each vertex in $B'$ corresponds to a vertex in the vertex set $B$.
   And a vertex $a \in A'$ is adjacent to a vertex $b \in B'$ if and only if the component in $G[A]$ corresponding to $a$ is adjacent to the vertex $b$ in $G$.
   Note that $2|A'| = |A|$ and $|B'| = |B|$. If $|A| > 2|B|$, then $|A'| > |B'|$.

   Since the induced subgraph $G[A]$ is an induced matching, $G'[A']$ is an independent set.
   Hence $G'$ has a VC crown decomposition $(C', H', R')$ with $C'\subseteq A'$ that can be found in linear time by Lemma \ref*{VC-lemma}.

   Associate $(C', H', R')$ to the decomposition $(C, H, R)$ of $V$, where $C$ is the components corresponding to the vertices in $C'$, $H$ is the vertices corresponding to the vertices in $H'$, $R = V \setminus (H \cup C)$.
   Let us check the three conditions in the definition of AIM crown decomposition in $(C, H, R)$.
   First, there is no edge between $C$ and $R$ since there is no edge between $C'$ and $R'$. Second, the induced subgraph $G[C]$ is an induced matching since  $C'$ is an independent set.
   Last, note that there is an injective mapping (matching) $M': H'\rightarrow C'$ such that, $\forall x\in H', \{x, M'(x)\}\in E'$. In the corresponding graph $G$, there is an injective mapping $M: H \rightarrow E(C)$ such that for all $v\in H$, there exists $u\in M(v)$ such that $(u,v)\in E$.

   Since $C'\neq \emptyset$, we can see $C \neq \emptyset$.
   So $(C, H, R)$ is an AIM crown decomposition of $V$ with $\emptyset \neq C \subseteq A$ and $H \subseteq B$, and this decomposition can be found in polynomial time.
\end{proof}

If the number of $Q_1$-edges only adjacent to vertices of good 3-paths in $\mathcal{P}$ is large,
we use Lemma \ref*{AIM-crown-lemma} and Lemma \ref*{AIM-crown-lemma-2} to reduce the instance.
Our algorithm first finds two vertex-disjoint sets of vertices, $A$ and $B$, which satisfy the condition in Lemma \ref*{AIM-crown-lemma-2} based on a proper $P_3$-packing $\mathcal{P}$.
Let $A$ be the set of $Q_1$-vertices that are only adjacent to good 3-paths in $\mathcal{P}$.
Let $B$ be the set of vertices in good 3-paths that are adjacent to some vertices in $A$.
If $|A| > 2|B|$, we can find an AIM crown decomposition $(C, H, R)$ by Lemma \ref*{AIM-crown-lemma-2},
and we can reduce the instance by including $C$ to the deletion set and removing $C \cup H$ from the graph.

Our algorithm, denoted by ${\tt Reduce}(G, k)$, is described in Fig. \ref{fig_algorithm}.

\begin{figure}[!t]
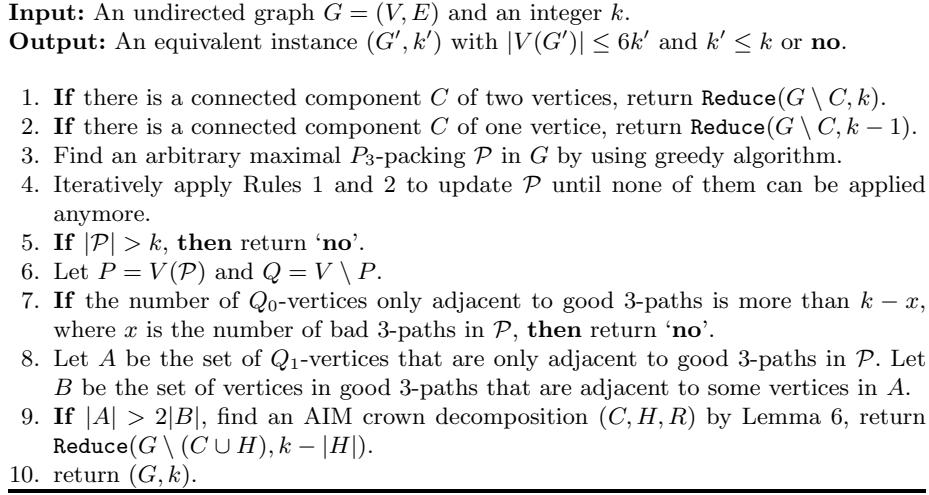

    \noindent \rule{\linewidth}{0.4mm}\\
    {\bf Input:} An undirected graph $G=(V,E)$ and an integer $k$.\\
    {\bf Output:}
    An equivalent instance $(G', k')$ with $|V(G')|\leq 6k'$ and $k'\leq k$ or \textbf{no}.

    \begin{enumerate}
        \item \textbf{If} there is a connected component $C$ of two vertices, return ${\tt Reduce}(G\setminus C, k)$.

        \item \textbf{If} there is a connected component $C$ of one vertice, return ${\tt Reduce}(G\setminus C, k - 1)$.


        \item Find an arbitrary maximal $P_3$-packing $\mathcal{P}$ in $G$ by using greedy algorithm.


        \item Iteratively apply Rules~1 and 2 to update $\mathcal{P}$ until none of them can be applied anymore.
        \item \textbf{If} $|\mathcal{P}| > k$, \textbf{then}  return `\textbf{no}'.
        \item Let $P=V(\mathcal{P})$ and $Q=V\setminus P$.
        \item \textbf{If} the number of $Q_0$-vertices only adjacent to good 3-paths is more than $k-x$,
        where $x$ is the number of bad 3-paths in $\mathcal{P}$, \textbf{then}  return `\textbf{no}'.
        \item Let $A$ be the set of $Q_1$-vertices that are only adjacent to good 3-paths in $\mathcal{P}$.
        Let $B$ be the set of vertices in good 3-paths that are adjacent to some vertices in $A$.
        \item \textbf{If} $|A| > 2|B|$, find an AIM crown decomposition $(C,H,R)$ by Lemma \ref*{AIM-crown-lemma-2},
        return ${\tt Reduce}(G\setminus (C\cup H),k-|H|)$.
        \item return $(G, k)$.
    \end{enumerate}

    \vspace{-6mm}

    \rule{\linewidth}{0.4mm}
    \caption{Algorithm ${\tt Reduce}(G,k)$}\label{fig_algorithm}
\end{figure}

\begin{lemma}\label{correctness}
    Algorithm ${\tt Reduce}(G, k)$ runs in polynomial time,
    and returns either an equivalent instance $(G', k')$ with $|V(G')|\leq 6k'$ and $k'\leq k$ or \textbf{no} to indicate that the instance is a \textbf{no}-instance.
\end{lemma}

\begin{proof}
    First, let us consider the correctness of each step.
    Steps 1-6 are trivial cases.
    Step 7 is based on Lemma \ref*{prop2} and
    Step 9 is based on Lemma \ref*{AIM-crown-lemma}.
    Step 8 is also trivial. Next, we consider Step 10.
    In this step, $\mathcal{P}$ is a proper $P_3$-packing and then the number of vertices in $P$ is at most $3k$.
    Assume that the number of bad 3-paths in $\mathcal{P}$ is $x$ and the number of good 3-paths in $\mathcal{P}$ is $y$.
    By the definition of proper $P_3$-packing, we know that the number of vertices in bad 3-paths and in $Q$-components adjacent to some bad 3-paths is at most $6x$.
    The number of vertices in good 3-paths is $3y$, the number of $Q_0$ -vertices only adjacent to good 3-paths is at most $k - x$ by Lemma \ref*{prop2}, and the number of $Q_1$-vertices only adjacent to good 3-paths is at most $2y$ by Lemma \ref*{AIM-crown-lemma} and Lemma \ref*{AIM-crown-lemma-2}.
    In total, the number of vertices in the graph is at most $ 6x + 3y + (k - x) + 2y = k + 5x + 5y \leq 6k$.
    Thus, we can get that $|V| \leq 6k$ in Step 10.

    Each step in ${\tt Reduce}(G, k)$ runs in polynomial time.
    Since each recursive call of ${\tt Reduce}(G, k)$ decreases $|V|$ by at least 1,
    ${\tt Reduce}(G, k)$ will be called at most $|V|$ times.
    Thus, ${\tt Reduce}(G, k)$ runs in polynomial time.
\end{proof}



The following Theorem 1 directly follows from Lemma~\ref*{correctness}.

\begin{theorem}
    \textsc{Almost Induced Matching} admits a kernel with $6k$ vertices.
\end{theorem}

\section{A parameterized algorithm}

In this section, we design a parameterized algorithm for \textsc{Almost Induced Matching}, which is a branch-and-search algorithm.
We use a parameter $k$ to measure the instance, and $T(k)$ to denote the maximum size of the search tree generated by the algorithm when running on an instance with the parameter no greater than $k$.
Assume that a branching operation generates $l$ branches and the measure $k$ in the $i$-th instance decreases by at least $c_i$. This operation generates a recurrence relation
\[
    T(k)\leq T(k-c_1)+T(k-c_2)+...+T(k-c_l)+1.
\]

The largest root of the function $f(x)=1-\sum_{i=1}^lx^{-c_i}$ is called the \textit{branching factor} of the recurrence.
Let $\gamma$ denote the maximum branching factor among all branching factors in the search tree.
The running time of the algorithm is bounded by $O^*(\gamma^k)$~\cite{kratsch2010exact}.

\subsection{Branching rules}

We have several branching rules that will be applied in different steps.


\vspace{2mm}
\noindent\textbf{Branching-Rule (B1).}
\textit{
    Branch on $v$ to generate $|N[v]|$ branches by
    either (i) deleting $v$ from the graph and including it in the deletion set,
    or (ii) for each neighbor $u$ of $v$,
    deleting $N[\{u, v\}]$ from the graph and including $N(\{u, v\})$ in the deletion set.
}
\vspace{2mm}

When dealing with certain graph structures, we can use a more effective branching rule.
A vertex $v$ \textit{dominates} its neighbor $u$ if $N[u] \subseteq N[v]$.
A vertex $v$ is called a \textit{dominating} vertex if it dominates at least one vertex.
The following property of dominating vertices has been used in \cite{gupta2012maximum,xiao2017exact}.

\begin{lemma}\label{dominating-lemma}
    Let $v$ be a vertex that dominates a vertex $u$. If there is a maximum induced matching $M$ of $G$ such that $v \in V (M)$, then there is a maximum induced matching $M'$ of $G$ such that edge $vu \in M'$.
\end{lemma}

We can use Lemma~\ref{dominating-lemma} to design an effective branching rule.

\vspace{2mm}
\noindent\textbf{Branching-Rule (B2).}
\textit{
    Assume that vertex $v$ dominates vertex $u$.
    Branch on $v$ to generate two instances by either (i) deleting $v$ from the graph and including it in the deletion set, or (ii) deleting $N[\{u, v\}]$ from the graph and including $N(\{u, v\})$ in the deletion set.
}
\vspace{2mm}

If there is a degree-2 vertex $v$ in a triangle, a more effective branching rule can be applied.
The following lemma has been used in \cite{xiao2017exact}.

\begin{lemma}\label{degree-2-lemma}
    If there is a degree-2 vertex $v$ in a triangle $vu_1u_2$,
    then there is a maximum induced matching either containing one edge in $\{vu_1, vu_2\}$ or containing no edge incident on a vertex in $\{v, u_1, u_2\}$.
    Especially, if at least one vertex in $\{u_1, u_2\}$ is of degree at least 3, then there is a maximum induced matching containing one edge in $\{vu_1, vu_2\}$.
\end{lemma}

We can use Lemma~\ref{degree-2-lemma} to design an effective branching rule to deal with degree-2 vertices in triangles.

\vspace{2mm}
\noindent\textbf{Branching-Rule (B3).}
\textit{
    Branch on a degree-2 vertex $v$ in a triangle with two neighbors $u_1$ and $u_2$ as follows \\
    (i) if $d(u_1)\leq 3$ or $d(u_2) \leq 3$, then generate two instances by either (a) deleting $N[\{v, u_1\}]$ from the graph and including $N(\{v, u_1\})$ in the deletion set,
    or (b) deleting $N[\{v, u_2\}]$ from the graph and including $N(\{v, u_2\})$ in the deletion set; \\
    (ii) if $d(u_1)\geq 4$ and $d(u_2) \geq 4$, then generate three instances by
    either (a) deleting $\{v, u_1, u_2\}$ from the graph and including them in the deletion set,
    (b) deleting $N[\{v, u_1\}]$ from the graph and including $N(\{v, u_1\})$ in the deletion set,
    or (c) deleting $N[\{v, u_2\}]$ from the graph and including $N(\{v, u_2\})$ in the deletion set.
}

\subsection{The algorithm}

We will use ${\tt aim}(G,k)$ to denote our parameterized algorithm.
Before executing the main branching steps, the algorithm will first apply some reduction rules to simplify the instance.
First of all, we call the kernel algorithm ${\tt Reduce}(G, k)$ to reduce the instance, which can be considered as Reduction-Rule 0.
We also have three more reduction rules.

\begin{rrule}\label{rrule-4}
    If there is a connected component of the graph such that each vertex in it is a degree-2 vertex,
    then select an arbitrary vertex $v$ in this component and return ${\tt aim}(G\setminus \{v\},k-1)$.
\end{rrule}

\begin{rrule}\label{rrule-5}
    If there is a degree-1 vertex $v$ with a degree-2 neighbor $u$,
    then return ${\tt aim}(G\setminus N[\{v,u\}],k-1)$.
\end{rrule}

Reduction-Rule~\ref{rrule-4} is trivial. Reduction-Rule~\ref{rrule-5}'s correctness is based on the observation: there is always a maximum induced matching containing edge $\{v, u\}$.

A cycle $u_0u_1u_2u_3$ of four vertices is called a \textit{short cycle} if the two vertices $u_0$ and $u_3$ are of degree at least 2 and the two vertices $u_1$ and $u_2$ are of degree 2.

\begin{lemma}\label{short-cycle-lemma}
    If a graph $G$ has a short cycle $u_0u_1u_2u_3$, then there is a maximum induced matching of $G$ containing the edge $u_1u_2$.
\end{lemma}

\begin{proof}
    Let $i\in \{0,3\}$ and $i'=\{0,3\}\setminus \{i\}$.
    If no edge incident on $u_i$ is in a maximum induced matching $M$, then by the maximality of $M$ we know that $u_1u_2$ is in $M$.
    Next, we assume that an edge $e$ incident on $u_i$ is in a maximum induced matching $M$.
    If $e=u_0u_3$, then we can replace $e$ with $u_1u_2$ in $M$.
    If $e\neq u_0u_3$, then edge $u_1u_2$ is not in $M$ and $u_{i'}\not\in V(M)$. For this case,
    we can also replace $e$ with $u_1u_2$ in $M$ to get a maximum induced matching containing $u_1u_2$.
    There is always a maximum induced matching containing $u_1u_2$.
\end{proof}

\begin{rrule}\label{rrule-6}
    If there is a short cycle $u_0u_1u_2u_3$, then return ${\tt aim}(G \setminus \{u_0,u_1,u_2,u_3\},k - 2)$.
\end{rrule}

Now, we are ready to introduce the main branching steps of the algorithm.
The algorithm contains nine steps to handle different local structures of the graph.
Dominating vertices are processed in Step 1, while Steps 2 to 5 are dedicated to handling degree-2 vertices.
Vertices with at least five neighbors are handled in Step 6,
and the last three steps focus on graphs with only degree-3/4 vertices.
When executing a step, we assume that all previous steps are not applicable to the current graph.

\vspace{2mm}

\noindent\textbf{Step 1} (Dominating vertices of degree at least 3).
    If there is a vertex $v$ of degree at least 3 that dominates a vertex $u$, then branch on $v$ with Rule (B2) to generate two branches
\[
    {\tt aim}(G \setminus \{v\},k-1) \quad \mbox{and} \quad {\tt aim}(G \setminus N[\{v,u\}],k-|N(\{v,u\})|).
\]

Lemma \ref*{dominating-lemma} guarantees the correctness of this step.
Note that $|N(\{v, u\})|$ = $d(v) - 1$.
This step generates a recurrence
\[
    T(k) \leq T(k - 1) + T(k - (d(v) - 1)) + 1,
\]
where $d(v) \ge 3$.
For the worst case where $d(v) = 3$, the branching factor of it is 1.6181.

After Reduction-Rule 2, degree-1 vertices can only be adjacent to vertices of degree at least 3.
These vertices will be handled in Step 1.
So after Step 1, there are no degree-1 vertices in $G$.

Next, we consider degree-2 vertices.
A path $u_0u_1u_2u_3u_4$ of five vertices is called a \textit{chain} if the first vertex $u_0$ is of degree at least 3 and the three middle vertices are of degree 2, where we allow $u_4$ = $u_0$.
A path $u_0u_1u_2u_3$ of four vertices is called a \textit{short chain} if the first vertex $u_0$ and last vertex $u_3$ are of degree at least 3 and the two middle vertices are of degree 2, where we allow $u_3$ = $u_0$.
A chain or a short chain can be found in linear time if it exists.
A short chain $u_0u_1u_2u_3$ is called a \textit{good short chain} if $u_3\neq u_0$ and there is no edge between $u_3$ and $u_0$.

\begin{lemma}\label{good-short-chain-lemma}
    If a graph $G$ has a good short chain $u_0u_1u_2u_3$, then there is a maximum induced matching of $G$ containing at least one vertex of $u_1$ and $u_2$.
\end{lemma}

\begin{proof}
    If none of $u_0$ and $u_3$ appears in a maximum induced matching $M$, then by the maximality of $M$, we know that $u_1u_2$ is in $M$.
    If only one of $u_0$ and $u_3$, say $u_0$, appears in the maximum induced matching $M$, we can always replace the edge incident on $u_0$ in $M$ with edge $u_0u_1$ to get another maximum induced matching containing $u_1$.
    If both $u_0$ and $u_3$ appear in a maximum induced matching $M$, while both $u_1$ and $u_2$ are not,
    we can also replace the edge incident on $u_0$ in $M$ with edge $u_0u_1$ to get another maximum induced matching.
    Consequently, there is a maximum induced matching containing at least one vertex in $u_1$ and $u_2$.
\end{proof}

\noindent\textbf{Step 2} (Chains).
    If there is a chain $u_0u_1u_2u_3u_4$, then branch on $u_1$ with Rule (B1).
    In the branch where $u_1$ is deleted and included in the deletion set, we get a tail $u_2$ and then further deal with the tail as we do in Reduction-Rule 2.
    Then we get the following three branches
    \[
        \begin{split}
        {\tt aim}(G \setminus \{u_1, u_2, u_3, u_4\},k-2),&\quad {\tt aim}(G\setminus N[\{u_0,u_1\}],k-|N(\{u_0,u_1\})|),\\
          \mbox{and}  & \quad {\tt aim}(G\setminus N[\{u_1,u_2\}],k-|N(\{u_1,u_2\})|).
        \end{split}
    \]

The corresponding recurrence is

\[
    T (k) \leq T (k - 2) + T (k - d(u_0)) + T (k - 2) + 1, %
\]
where $d(u_0) \ge 3$.
For the worst case where $d(u_0) = 3$, the branching factor of it is 1.6181.

After Step 1, there is no short chain $u_0u_1u_2u_3$ with $u_0 = u_3$, since for any short chain $u_0u_1u_2u_3$ with $u_0 = u_3$, vertex $u_0$ is a dominating vertex, and Step 1 would be applied.
After Reduction-Rule 3, no short chain $u_0u_1u_2u_3$ exists with an edge $u_0u_3$ in $G$.
Therefore, we claim that every short chain in $G$ is a good short chain after Step 1.

  \vspace{2mm}
\noindent\textbf{Step 3} (Short chains).
If there is a short chain $u_0u_1u_2u_3$, then branch on $u_1$ with Rule (B1). Additionally, by Lemma \ref*{good-short-chain-lemma}, in the branch of deleting $u_1$, we can delete $N[\{u_2, u_3\}]$ from the graph and include $N(\{u_2, u_3\})$ in the deletion set.
Thus, we get the following three branches
\[
    \begin{split}
        {\tt aim}(G\setminus N[\{u_2,u_3\}],k-|N(\{u_2,u_3\})|), &\quad {\tt aim}(G\setminus N[\{u_0,u_1\}],k-|N(\{u_0,u_1\})|),\\
        \mbox{and}  & \quad {\tt aim}(G\setminus N[\{u_1,u_2\}],k-|N(\{u_1,u_2\})|).
    \end{split}
\]

Note that $|N(\{u_2, u_3\})| = d(u_2) +d(u_3) - 2 = d(u_3), |N(\{u_0, u_1\})| = d(u_0) +d(u_1) - 2 = d(u_0) \mbox { and } |N(\{u_1, u_2\})| = d(u_1) + d(u_2) - 2 = 2$.
The corresponding recurrence is
\[
  T (k) \leq T (k - d(u_3)) + T (k - d(u_0)) + T (k - 2) + 1,
\]
where $d(u_0)$ and $d(u_3) \ge 3$.
For the worst case where $d(u_0) = d(u_3) = 3$, the branching factor is 1.5214.

After Step 3, each degree-2 vertex in the graph has two neighbors of degree at least 3.
Since there is no dominating vertex after Step 1, the two neighbors of any degree-2 vertex are not adjacent to each other.

\vspace{2mm}
\noindent\textbf{Step 4} (Degree-2 vertices adjacent to a vertex of degree at least 4).
    If there is a degree-2 vertex $v$ adjacent to a vertex $u_1$ of degree at least 4,
    then branch on $v$ with Rule (B1) to generate $d(v) + 1$ branches
    \[
        {\tt aim}(G \setminus \{v\},k - 1) \quad \mbox{and} \quad {\tt aim}(G \setminus N[\{v, u\}],k - |N(\{v, u\})|) \mbox{~for each $u\in N(v)$}.
    \]

Let $u_2$ denote the other neighbor of $v$.
Then $u_2$ is not adjacent to $u_1$ and $d(u_2) \geq 3$.
The branching operation will generate three branches. Since $u_1$ and $u_2$ are not adjacent, we can see that $|N(\{v, u_1\})| = d(u_1) + d(v) - 2 = d(u_1) \ge 4$ and $|N(\{v, u_2\})| = d(u_2) + d(v) - 2 = d(u_2) \ge 3$.
This leads to a recurrence

\[
    T (k) \leq T (k - 1) + T (k - d(u_1)) + T (k - d(u_2)) + 1,
\]
where $d(u_1) \geq 4$ and $d(u_2) \geq 3$.
For the worst case where $d(u_1) = 4$ and $d(u_2) = 3$, the branching factor is 1.6181.
\vspace{2mm}

After Step 4, if there is a degree-2 vertex $v$ adjacent to two vertices $u_1$ and $u_2$, then $d(u_1)\geq 3$ and $d(u_2) \geq 3$.
Vertex $u_1$ is not adjacent to $u_2$ since there is no dominated vertex.
Let $N(u_1)=\{v,w_1,w_2\}$ and $N(u_2)=\{v,w_3,w_4\}$, where it is possible that $\{w_1,w_2\}\cap \{w_3,w_4\} \neq \emptyset$.
In Step 5, we are going to deal with such degree vertices $v$.

\vspace{2mm}
\noindent\textbf{Step 5} (Degree-2 vertices with two nonadjacent degree-3 neighbors).
    We deal with such degree-2 vertices $v$ by considering two different cases.

    Case 1. $|\{w_1,w_2\}\cap \{w_3,w_4\}| \geq 1$: We branch on $w_1$ with Rule (B1) to generate $d(w_1) + 1$ branches.

    \begin{figure}[t]
        \centering
        \includegraphics[scale=0.2]{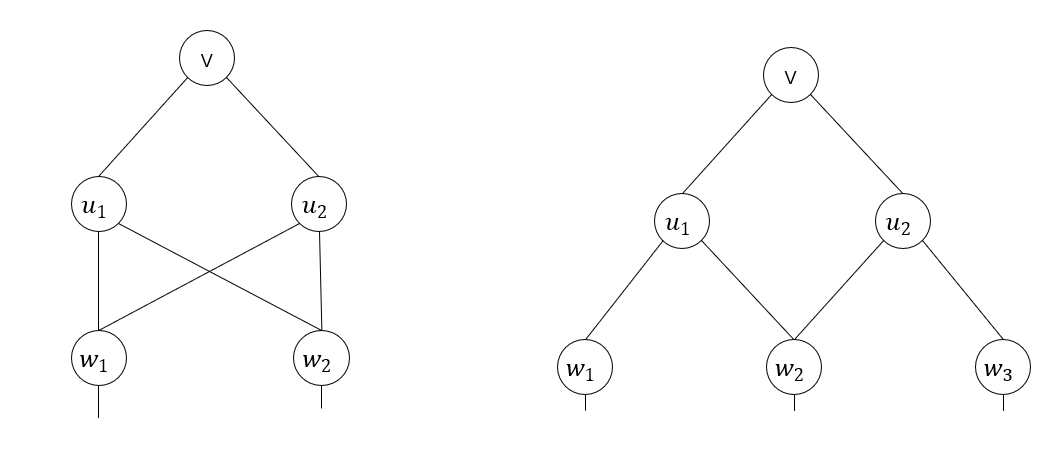}
        \caption{Step 5. Case 1.}
        \label{Fig:Step-5-Case-1}
    \end{figure}

    For this case, we assume without loss of generality $w_2 = w_4$.
    See Fig. \ref*{Fig:Step-5-Case-1} for an illustration.
    In the branch where $w_1$ is deleted and included in the deletion set, we get a short cycle $w_2u_1vu_2$ and then further deal with the short cycle as we do in Reduction-Rule 3.
    Then we get the following three branches
    \[
        \begin{split}
        & {\tt aim}(G \setminus \{w_1, w_2, u_1, v, u_2\},k - 3) \\
        \mbox{and} \quad & {\tt aim}(G \setminus N[\{w_1, x\}],k - |N(\{w_1, x\})|) \mbox{~for each $x\in N(w_1)$}.
        \end{split}
    \]

    If $d(w_1) = 2$, let $x_1$ and $x_2$ denote the two neighbors of $w_1$.
    We know that $x_1$ and $x_2$ are nonadjacent vertices of degree at least 3. So $w_1$ is not a dominating vertex.
    If $d(w_1) \geq 3$, according to Step 1, we know that $w_1$ is not a dominating vertex.
    Thus, each neighbor $x$ of $w_1$ is adjacent to at least one vertex out of $N[w_1]$ and then $|N(\{w_1, x\})| \ge d(w_1)$.
    We get a recurrence
    \[
        T (k) \leq T (k - 3) + d(w_1)\times T (k - d(w_1)) + 1,
    \]
    where $d(w_1) \geq 2$.
    For the worst case where $d(w_1) = 2$, the branching factor is 1.6181.

    Case 2. $|\{w_1,w_2\}\cap \{w_3,w_4\}| = 0$:
    Firstly, if $N(\{v,u_1,u_2,w_1,w_2,w_3,w_4\})=\emptyset$, we delete $\{v,u_1,u_2,w_1,w_2,w_3,w_4\}$ from the graph and include $\{v,w_2,w_4\}$ in the deletion set. The correctness of this operation can be easily verified.
    Next, we assume without loss of generality that $N(\{v,u_1,u_2,w_1,w_2,w_3,w_4\})\neq \emptyset$. Let $x^*$ be a vertex adjacent to $w_1$.
    See Fig. \ref*{Fig:Step-5-Case-2} for an illustration. One of the following three cases will happen.

    \begin{figure}[t]
        \centering
        \includegraphics[scale=0.25]{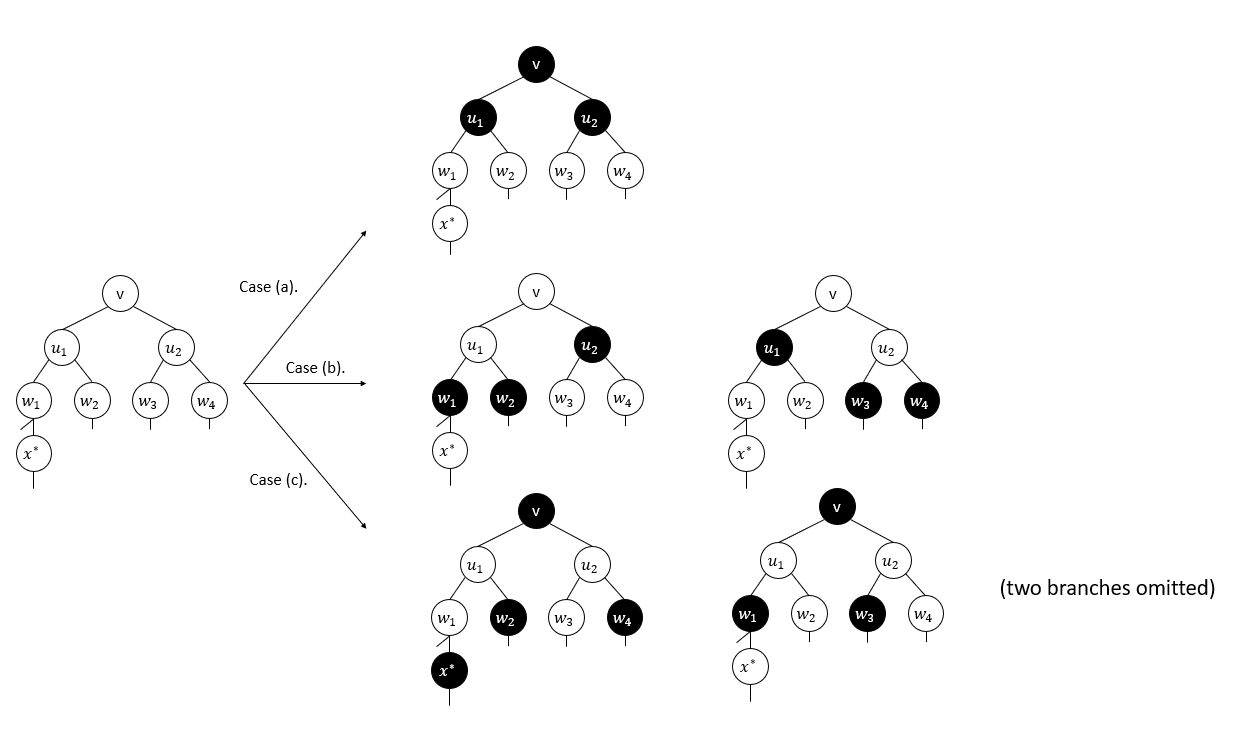}
        \caption{Step 5. Case 2: Vertices in the deletion set are denoted by black vertices.}
        \label{Fig:Step-5-Case-2}
    \end{figure}

    Case (a). Both of $u_1$ and $u_2$ are deleted and included in the deletion set:
    We delete $\{v, u_1, u_2\}$ from the graph and include them in the deletion set. The parameter $k$ decreases by 3.

    Case (b). Only one of $u_1$ and $u_2$ is deleted and included in the deletion set. We use $u_i (i \in \{1, 2\})$ to denote the other vertex:
    We can see that there is a maximum induced matching containing edge $u_i v$.
    Then we delete $N[\{u_i, v\}]$ from the graph and including $N(\{u_i, v\})$ in the deletion set. The parameter $k$ decreases by at least 3.

    Case (c). Neither $u_1$ nor $u_2$ is included in the deletion set:
    Since $u_1$ and $u_2$ are not adjacent, they must be in two different edges in the maximum induced matching $M$.
    We first generate two branches by deleting $N[\{u_1, w_i\}]$ for $i = 1$ and $2$ and including $N(\{u_1, w_i\})$ in the deletion set. The resulting graph is denoted as $G_i$.
    In $G_i$, since $u_2$ appears in $M$, vertex $u_2$ is not deleted or becomes a degree-0 vertex.
    If $u_2$ has only one neighbor $v^*$ or dominates one of its neighbors $v^*$, we further delete $N[\{u_2, v^*\}]$ and include $N(\{u_2, v^*\})$ in the deletion set.
    Otherwise we further branch into two branches by deleting $N[\{u_2, w_j\}]$ from the graph and including $N(\{u_2, w_j\})$ in the deletion set for $j = 3$ and $4$.
    In each branch, the parameter $k$ decreases by at least 2 since $u_2$ is not a dominating vertex and $w_j$ is adjacent to a vertex not in $N[u_2]$. 

    Note that in the subbranch by deleting $N[\{u_1, w_1\}]$, vertex $x^*$ is included in the deletion set and there are at least 3 vertices in $\{u,v_1,v_2,w_1,w_2,w_3,w_4\}$ are included in the deletion set.
    As a result, the above branching operation will generate a recurrence covered by one of the following three options:

    \[
        \begin{split}
            & T (k) \leq 3T (k - 3) + T(k - 3) + T(k - 4) + 1,\\
            & T (k) \leq 3T (k - 3) + T(k - 3) + 2T(k - 5) + 1,
        \end{split}
    \]
    and\\
    \[
        T (k) \leq 3T (k - 3) + 4T(k - 5) + 1.\\
    \]

    The branching factors of them are 1.6634, 1.6765 and 1.6478, respectively.

\vspace{2mm}
\noindent\textbf{Step 6} (Vertices of degree at least 5).
If there is a vertex $v$ of $d(v) > 4$, then branch on $v$ with Rule (B1) to generate $d(v) + 1$ branches
\[
    {\tt aim}(G \setminus \{v\},k - 1) \quad \mbox{and} \quad {\tt aim}(G \setminus N[\{v, u\}],k - |N(\{v, u\})|) \mbox{~for each $u\in N(v)$}.
\]

Since $v$ is a vertex of degree at least 5, due to Step 1, we know that it can not dominate any neighbor of it.
Thus each neighbor $u$ of $v$ is adjacent to at least one vertex out of $N[v]$ and then $|N(\{v, u\})| \ge d(v)$.
So we get a recurrence
\[
    T (k) \leq T (k - 1) + d(v)\times T (k - d(v)) + 1,
\]
where d(v) $> 4$.
For the worst case where $d(v) = 5$, the branching factor of it is 1.6595.
\vspace{2mm}





After the first six steps, the graph contains only vertices with degree 3 and 4, and there is no dominating vertex.
The next three steps are used to deal with degree-3/4 vertices in the graph.
For any degree-3 vertex $v$, there is at most one edge between its neighbors,
otherwise there must be a dominating vertex in the neighbors of $v$, and Step 1 would be applied.
We first consider a degree-3 vertex $v$ not contained in any triangle.

\vspace{2mm}
\noindent\textbf{Step 7} (Degree-3 vertices not in any triangle).
    If there is a degree-3 vertex $v$ such that there is no edge between its neighbors,
    then branch on $v$ with Rule (B1) to generate $d(v) + 1$ branches
    \[
        {\tt aim}(G \setminus \{v\},k - 1) \quad \mbox{and} \quad {\tt aim}(G \setminus N[\{v, u\}],k - |N(\{v, u\})|) \mbox{~for each $u\in N(v)$}.
    \]



Let $u_1$, $u_2$ and $u_3$ denote the three neighbors of $v$.
We have that $u_1,u_2,u_3$ are nonadjacent vertices of degree at least 3.
The branching operation will generate three branches.
Since $u_1$, $u_2$ and $u_3$ are not adjacent, we can see that $|N(\{v, u_i\})| = d(u_i) + d(v) - 2 = d(u_i) + 1 \ge 4 (i = 1, 2, 3)$.
This leads to a recurrence
\[
    T (k) \leq T (k - 1) + T (k - (d(u_1) + 1)) + T (k - (d(u_2) + 1)) + T (k - (d(u_3) + 1)) + 1,
\]
where $d(u_1),d(u_2),d(u_3) \ge 3$.
For the worst case that $d(u_1) = d(u_2) = d(u_3) = 3$, the branching factor is 1.6581.

Next we consider degree-3 vertices in triangles that are also adjacent to some degree-4 vertex.


\vspace{2mm}
\noindent\textbf{Step 8} (Degree-3 vertices adjacent to some degree-4 vertex).
    Assume there is a degree-3 vertex $v$ adjacent to at least one degree-4 vertex. After Step 7, each degree-3 vertex is in exactly one triangle since such degree-3 vertex can not be dominated.
    Let $u_1$, $u_2$ and $u_3$ be the three neighbors of $v$, where we assume without loss of generality that there is an edge between $u_1$ and $u_2$.
    First, we branch on $u_3$ with Rule (B1). In the branch of deleting $u_3$, we further branch on the vertex $v$ with Rule (B3).

Note that the degree-4 neighbor of $v$ can be any one of $u_1,u_2$ and $u_3$.
\begin{lemma}\label{Step 8-lemma}
    The branching factor of Step 8 is at most 1.6430.
\end{lemma}

\begin{proof}

    First, we assume without loss of generality that $d(u_1) \geq d(u_2)$.
    See Fig. \ref*{Fig:3} for an illustration.
    \begin{figure}[t]
        \centering
        \includegraphics[scale=0.1]{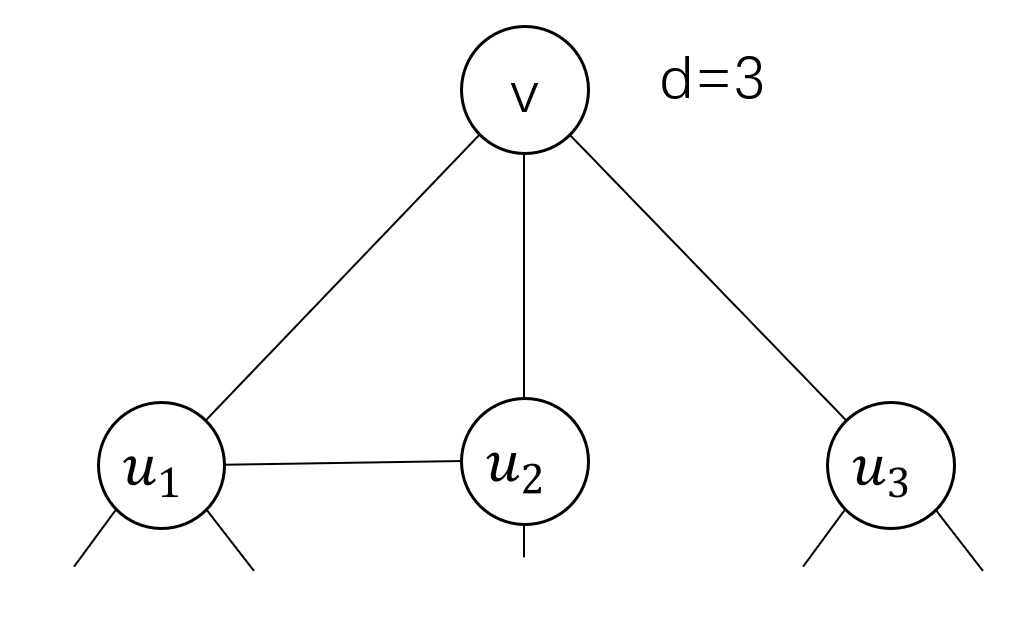}
        \caption{A degree-3 vertex $v$ in a triangle}
        \label{Fig:3}
    \end{figure}

Vertex $u_3$ is a vertex of degree at least 3 not dominating any neighbor of it.
Thus each neighbor $w$ of $u_3$ is adjacent to at least one vertex out of $N[u_3]$ and then $|N(\{u_3, w\})| \ge d(u_3)$.
Since there is an edge between $u_1$ and $u_2$, and both of $u_1$ and $u_2$ do not dominate $v$, vertex $v$ is adjacent to two vertices out of $N[u_3]$ and then $|N(\{u_3, v\})| = d(u_3) + 1$.
So we get a recurrence
\[
    T (k) \leq T (k - 1) + T (k - d(u_3) - 1) + (d(u_3) - 1) \times T (k - d(u_3)) + 1.
\]

Then, we consider two cases.

\textbf{Case 1.} $d(u_1) = d(u_2) = 4$: In the branch of deleting $u_3$, we further branch on the vertex $v$ with Rule (B3).
We get three branches
\[
    \begin{split}
    {\tt aim}(G \setminus \{v, u_1, u_2\}, k - 3), & \quad {\tt aim}(G \setminus N[\{u_1, v\}],k - |N(\{u_1, v\})|)\\
    \mbox{and} & \quad {\tt aim}(G \setminus N[\{u_2, v\}],k - |N(\{u_2, v\})|).
    \end{split}
\]

This leads to a recurrence
\[
    T(k) \leq T(k - 3) + T(k - (4 - 1)) + T(k - (4 - 1)) + 1.
\]

After combining the two branching operations, we get
\[
    \begin{split}
        {\tt aim}(G \setminus \{v, u_1, u_2, u_3\},k - 4), \quad {\tt aim}(G \setminus N[\{u_1, v\}], k -|N(\{u_1, v\})| - 1)&,\\
        {\tt aim}(G \setminus N[\{u_2, v\}], k -|N(\{u_2, v\})| - 1)&,\\
        \mbox{and}  \quad {\tt aim}(G \setminus N[\{u_3, w\}],k - |N(\{u_3, w\})|) \mbox{~for each $w \in N(u_3)$}&.
    \end{split}
\]


The corresponding recurrence is
\[
    T(k) \leq T(k - 4) + 2T(k - 3) + T (k - (d(u_3) + 1)) + (d(u_3) - 1) \times T (k - d(u_3)) + 1,
\]
where $d(u_3) \ge 3$.
For the worst case where $d(u_3) = 3$, the branching factor is 1.6430.

\textbf{Case 2.} $d(u_2) = 3$: In the branch of deleting $u_3$, we further branch on the vertex $v$ with Rule (B3).
We get two branches
\[
    {\tt aim}(G \setminus N[\{u_1, v\}],k - |N(\{u_1, v\})|) \quad \mbox{and} \quad {\tt aim}(G \setminus N[\{u_2, v\}],k - |N(\{u_2, v\})|).
\]

This leads to a recurrence
\[
    T(k) \leq  T(k - (d(u_1) - 1)) + T(k - (3 - 1)) + 1.
\]

After combining two branching operations, we get
\[
    \begin{split}
        \quad {\tt aim}(G \setminus N[\{u_1, v\}], k -|N(\{u_1, v\})| - 1)&,\\
        {\tt aim}(G \setminus N[\{u_2, v\}], k -|N(\{u_2, v\})| - 1)&,\\
        \mbox{and}  \quad {\tt aim}(G \setminus N[\{u_3, w\}],k - |N(\{u_3, w\})|) \mbox{~for each $w \in N(u_3)$}&.
    \end{split}
\]

The corresponding recurrence is
\[
    T(k) \leq T(k - d(u_1)) + T(k - 3) + T (k - (d(u_3) + 1)) + (d(u_3) - 1) \times T (k - d(u_3)) + 1,
\]
where $\min \{d(u_1),d(u_3)\} \ge 3$ and $\max \{d(u_1),d(u_3)\} \ge 4$.

\begin{table}[!t]
    \begin{center}
    \caption{The branching factors for different cases in Step 8. Case 2.}
    \begin{tabular}{l|c}
        \hline
        degree & branching factor \\
        \hline
        $d(u_1)=4,d(u_3)=4$ & 1.5785 \\  
        $d(u_1)=3,d(u_3)=4$ & 1.6181\\
        $d(u_1)=4,d(u_3)=3$ & 1.6181 \\
        $d(u_1)=3,d(u_3)=3$ & - \\
        \hline
    \end{tabular}\label{tb-Step-8-upd}
    \end{center}
\end{table}

Table~\ref{tb-Step-8-upd} shows the branching factors for different cases according to the value of $d(u_1)$ and $d(u_3)$.
Note that at least one vertex in $\{u_1, u_3\}$ is a degree-4 vertex. In Case 2, the worst branching factor is 1.6181.
However, for Case 1, the branching factor is 1.6430. The worst branching factor in Step 8 remains at 1.6430.
\end{proof}

The worst branching factors in the above eight steps are listed in Table 1.
After Step 8, there are no degree-3 vertices adjacent to degree-4 vertices, and each connected component is either 3-regular or 4-regular.

\begin{table}[!t]\label{tb-result-L}
    \begin{center}
    \caption{The branching factors of each of the first eight steps}
    \begin{tabular}{l|c|c|c|c|c|c|c|c}
       \hline
       Steps & Step 1 & Step 2 & Step 3 & Step 4 & Step 5 & Step 6 & Step 7 & Step 8\\
       \hline
       branching factors & 1.6181 & 1.6181 & 1.5214 & 1.6181 & \textbf{1.6765} & 1.6595 & 1.6581 & 1.6430 \\
       \hline
    \end{tabular}
    \end{center}
    \vspace{-6mm}
\end{table}

\vspace{2mm}
\noindent\textbf{Step 9} (3/4-regular graphs).
    Pick up an arbitrary vertex $v$ in the 3/4-regular graph and branch on it with Rule (B1).
\vspace{2mm}

Step 9 will not exponentially increase the running time bound of the algorithm. We can prove the following theorem.

\begin{theorem}
    \textsc{Almost Induced Matching} can be solved in $O^*(1.6765^k)$ time and polynomial space.
\end{theorem}

\begin{proof}
   Among all the branching factors in the first eight steps, the worst one is 1.6765, which is in Step~5.
   If Step~9 is not executed, the algorithm ${\tt aim}(G, k)$ clearly runs in $O^*(1.6765^k)$ time.

   For a connected 3-regular graph, Step 9 can be applied for at most one time,
   as any properly induced subgraph of a connected 3-regular graph is neither a 3-regular nor a 4-regular graph. Thus, Step 9 will not affect the exponential part of the running time.
   We get that the algorithm runs in $O^*(1.6765^k)$  time on any connected 3-regular graph.
   For a 3-regular graph $G$, since we can solve each connected component in $O^*(1.6765^k)$ time and there are at most $n$ connected components, we can solve $G$ in $O^*(1.6765^k)$  time.

   For a connected 4-regular graph, after executing Step 9 for once (which will not exponentially affect the running time), the algorithm can always branch with a branching factor at most 1.6765 or
   get a 3-regular graph. For the latter case, we solve it in $O^*(1.6765^k)$  time directly by the above analysis. Thus, we can always solve a connected 4-regular graph in $O^*(1.6765^k)$ time.
   By the similar argument, if each connected component is 3-regular or 4-regular, we can solve the graph in $O^*(1.6765^k)$  time.
   This also implies that we can solve any graph in $O^*(1.6765^k)$ time.

   Each step uses only polynomial space. Thus, the theorem holds.
\end{proof}

\section{Conclusion}
In this paper, we study \textsc{Almost Induced Matching} from the prespective of parameterized algorithms, where the parameter $k$ represents the size of the deletion set.

In the context of kernelization, we introduce an enhanced structure called AIM crown decomposition, which effectively yields a $6k$-vertex kernel.
For further improvements, we may need to explore new structural properties and employ different techniques.
Note that the number of vertices in the kernel is already small.
It would also be interesting to achieve some nontrivial lower bounds for the kernel size.

In our parameterized algorithm, by using new methods to deal with degree-$3$ vertices adjacent to degree-$4$ vertices in the graph, we successfully avoid the bottlenecks in previous papers.
Table 2 shows the new bottleneck case in our algorithm generated by Step 5, which is to deal degree-$2$ vertices adjacent to two degree-$3$ vertices without an edge between them.

\section*{Acknowledgements}

This work was supported by the National Natural Science Foundation of China (Grant No. 62372095 and 62172077) and the Sichuan Natural Science Foundation (Grant No. 2023NSFSC0059).

%
%
\bibliographystyle{splncs04}
\bibliography{AIM}

\end{document}